\documentclass[12pt, draft, a4paper, onecolumn]{IEEEtran}
\usepackage[final]{graphicx}

\date{}
\usepackage{amsmath,epsfig,amssymb,verbatim,amsopn,cite,subfigure,multirow}
\usepackage{balance}
\usepackage[usenames,dvipsnames]{color}
\usepackage[all]{xy}  %%% used to make block diagram
\usepackage{url}
\usepackage{amsfonts}
\usepackage{amssymb}
\usepackage{epsfig}
\usepackage{bm}
\usepackage{epsfig}
\usepackage{amsmath,amssymb, amstext}
\usepackage{graphicx}
\usepackage{epstopdf}
\usepackage{algorithm}
\usepackage{algorithmic}
\usepackage{stackrel}

\usepackage{amsthm}

\usepackage{flushend}

\newtheorem{lemma}{Lemma}
\newtheorem{proof of lemma}{proof of Lemma}
\newtheorem{theo}{Theorem}
\makeatletter
\newcommand{\vast}{\bBigg@{3.5}}
\newcommand{\Vast}{\bBigg@{4.25}}
\makeatother
\begin{document}

\title{On the Capacity of Discrete-Time Laguerre Channel}
%\author{Hossein~Esmaeili
%and Jawad~A.~Salehi, \emph{Fellow, IEEE}
%\thanks{The authors are with the Optical Networks Research Laboratory (ONRL), Electrical Engineering Department, Sharif University of Technology, Tehran, Iran (e-mails: hossein\underline{ }esmaeili@ee.sharif.edu and jasalehi@sharif.edu). }}
\author{Hossein Esmaeili, and Jawad A. Salehi \emph{Fellow, IEEE} \\
%Optical Researchs Laboratory \\
%Electrical Engineering Department\\
%Sharif University of Technology\\
%(Emails: hossein \underline{ }esmaeili, jasalehi@sharif.edu)

\thanks{The authors are with the Optical Networks Research Laboratory (ONRL), Electrical Engineering Department, Sharif University of Technology, Tehran, Iran (e-mails: hossein\underline{ }esmaeili@ee.sharif.edu and jasalehi@sharif.edu). }% <-this % stops a space
}
%\markboth{XXX ,~Vol.~XX, No.~XX, MONTH~XXXX}
%{Jamali \MakeLowercase{\textit{et al.}}: XXX}
\maketitle

\begin{abstract}
In this paper, new upper and lower bounds are proposed for the capacity of discrete-time Laguerre channel. Laguerre behavior is used to model various types of optical systems and networks such as optical amplifiers, short distance visible light communication systems with direct detection and coherent code division multiple access (CDMA) networks. Bounds are derived for short distance visible light communication systems and coherent CDMA networks. These bounds are separated in three main cases: when both average and peak power constraints are imposed, when peak power constraint is inactive and when only peak power constraint is active.

%Implementation of cognitive radio networks as an opportunistic solution to overcome the scarcity of spectrum is a challenging problem. In order to opportunistically use a frequency band and without causing any interference with primary user, we need to specify the idle channel via spectrum sensing. If we choose energy detection as our method for spectrum sensing, then we need to have perfect knowledge about noise power, otherwise under noise uncertainty, performance of energy detection severely degrades. Also hidden terminal problem makes energy detection difficult specially in low signal-to-noise (SNR) regime. Thus we can mention noise uncertainty and hidden terminal problem as two main problems that do not allow us to perform an effective and efficient spectrum sensing. In this paper, we consider multiple cognitive radios cooperatively use an blind spectrum sensing framework in order to determine the occupancy of a channel. Because of using multiple cognitive radios instead of one, our algorithm maintains its
\end{abstract}
%\begin{IEEEkeywords}
%cognitive radio - spectrum sensing - energy detection - Bayesian estimation
%\end{IEEEkeywords}

\IEEEpeerreviewmaketitle

%------------------------------------<<intro>>-----------------------------------------
\section{Introduction}

\IEEEPARstart{O}{ptical}
 intensity modulation with direct detection (IM/DD) is one of the most prevalent methods to communicate through optical channels and networks due to its simplicity in design and implementation. In these channels, information is modulated onto intensity domain and thus, all symbols have non-negative values. To find the capacity of such channels, first one should obtain the statistical expression of the channel.  There has already been presented several channel statistics model for IM/DD channels such as Poisson and Gaussian intensity channels and also the capacity of these channels are investigated. In [1]-[5], upper and lower bound for discrete-time Poisson channel are proposed under different conditions. Upper and lower bounds for the capacity of the Guassian optical intensity channels are also evaluated in [6]-[10] using various methods such as sphere packing, duality approach and maxentropic method. Moser presents the capacity results of optical intensity channels with input-dependent gaussian noise under peak and average power conditions [11].

In this paper, Laguerre behavior is considered as the conditional statistics of the channel. Discrete-time Laguerre channel is an appropriate model for low-power optical communications when received signals contain monochromatic (single frequency) plus narrow band gaussian lightwaves which are related to input data and noise term respectively [12].
This model can be used for various optical communications such as free-space optical (FSO) communications, optical amplifiers and intersatellite laser links [13]. It is noteworthy that Laguerre behavior is an improved version of the poisson channel when the background noise factor is notable and it has been shown that the Poisson channel turns to the Laguerre one when the input of the Poisson channel is corrupted by a narrow band gaussian noise [12].

In discrete-time Laguerre channels, input data is a monochromatic lightwave and coding scheme is applied onto the  intensity of it, therefore the input data must be non-negative. It is important to note that since the input data is single frequency, optical input power has an direct effect on the photons rate which arrive to the receiver. But the energy of each incident photon is identical to all the others and is independent of the optical input power  and only the number of emitted photons are varied by variation of input power. At receiver, arrival signals are a stream of counted photons and thus, the outputs of such channels, unlike inputs, can give only non-negative discrete values.

In this work, first, upper and lower bounds for the capacity of the discrete-time Laguerre channel with independent noise are calculated when the average and peak constraints are imposed to the input power. Afterward, it is shown that optical coherent CDMA network statistics can be modeled as a Laguerre channel with input-dependent noise factor and then, some achievable rates are proposed for such a channel under different input constraints.

The rest of this paper is organized as follows. Section II gives system model of the optical channel whose statistics can be modeled as a Laruerre distribution, In Section III, our main results are proposed and In Section IV, derivation of the lower and upper bounds are presented.
%--------------------------------<< System Model >>------------------------------------
\section{System Model}
Considering many applications for short distance visible light communication systems and optical coherent CDMA networks, determining capacity of these channel models is the main key to specify maximum achievable data rates.

In this paper, a memoryless discrete-time channel is assumed in which its output and input are denoted by $Y$ and $X$ respectively in a way that $X \in \mathbb{R}^+$ and $Y \in \mathbb{Z}^+$. Note that, while input $X$ takes values from $\mathbb{R}^+$, the output $Y$ being a nonnegative integer is resulted from photo detector properties [12]. For coherent CDMA, short optical free space systems with direct detection and optical amplifiers, it can be shown that the channel statistics, as formulated in (\ref{laguerre1}), is Laguerre with mean $x+\lambda_0$ where $\lambda_0$ represents noise average power.

%--------------------------------------------<<eq1>>----------------------------------------------
\begin{equation}
\begin{split}
W(y|x)&=\frac{e^{\frac{-x}{\lambda}}}{1+\lambda}\left(\frac{\lambda}{1+\lambda} \right)^y \left( \sum_{j=0}^{\infty} \left( \frac{x}{\lambda (1+\lambda)} \right)^j \frac{(y+j)!}{(j!)^2 y!} \right)
\end{split}
\label{laguerre1}
\end{equation}
supposing that average power is a crucial factor, there is a limitation on maximum transmitted average power namely average power constraint as in (\ref{avecon}). In addition, considering that optical fibers and transmitters operate linearly in specific range of input power, another constraint -peak power constraint- can be imposed as (\ref{peakcon}),
%--------------------------------------------<<eq2>>----------------------------------------------
\begin{equation}
\begin{split}
E[X] \leq \mathcal{E}
\end{split}
\label{avecon}
\end{equation}
%--------------------------------------------<<eq3>>----------------------------------------------
\begin{equation}
\begin{split}
\Pr [X \geq A]=0,
\end{split}
\label{peakcon}
\end{equation}
where $A$ and $\mathcal{E}$ denote peak power and average power constraints respectively.

Here we define a parameter $\alpha$ as below
%--------------------------------------------<<eq4>>----------------------------------------------
\begin{equation}
\begin{split}
\alpha \triangleq \frac{\mathcal{E}}{A},
\end{split}
\label{alpha}
\end{equation}
where $\alpha$ takes values from $0^+$ to $1$. It is obvious that $\alpha=1$ represents the case in which average power constraint is inactive. On the other hand, for $\alpha \ll 1$, only average power constraint is taken into account. The channel capacity can be formulated as
%--------------------------------------------<<eq5>>----------------------------------------------
\begin{equation}
\begin{split}
\mathcal{C}(A,\mathcal{E})=\stackrel[Q(\cdot)]{}{\sup} I(X;Y),
\end{split}
\label{alpha}
\end{equation}
where $I(X;Y)$ represents the mutual information between $X$ and $Y$, where supremum is taken such that $\Pr [X \geq A]=0$, $E[X] \leq \mathcal{E}$ and $\Pr [X \leq A]=0$.

Similarly, $C(A)$ and $C(\mathcal{E})$ represent the cases in which peak power and average power constraints is taken into account singly.

Moreover, in coherent CDMA systems, average noise power depends on average input power. while there is not such dependency between these two parameters in short free space optical systems. So our results fall into two main aforementioned categories. Therefore in this paper, first set of simulations gives upper and lower bounds for short free space optical systems and in second set of simulations, lower bounds for coherent CDMA systems is achieved. Note that because of dependency between average noise and input power in latter systems, upper bound is not practical.
%%%%%%%%%%%%%%%%%%%%%%%%%%%%%%%%%%%%%%
\begin{figure}[t]
\centering
\includegraphics[scale=.35]{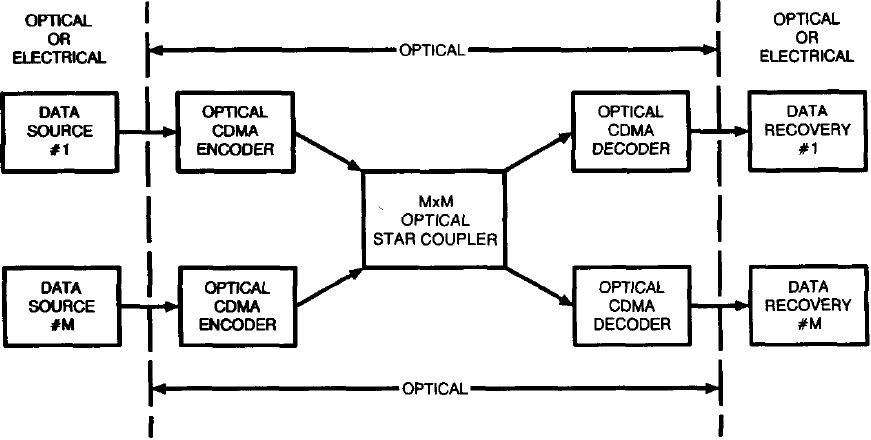}
\caption{A typical optical coherent CDMA system with an all-optical encoder/decoder.}
\label{first}
\end{figure}
%%%%%%%%%%%%%%%%%%%%%%%%%%%%%%%%%%%%%
In optical coherent CDMA systems, first, users encode their transmitting signal and then transmit it on the channel. 
After that, all users' data are added together and at the receiver end, each user multiplies its receiving signal 
to its corresponding code and reproduces transmitting data. As a result, our channel model is interfering. 
Fig.\ref{first} depicts a typical coherent optical CDMA network with $M$ users. As it is shown in this figure, first, 
each optical pulse is transferred to frequency domain by grating and then the transformed optical pulse is directed to 
the encoder phase mask and finally, the encoded pulse is transformed to time domain again. Since each user has a 
distinct encoder/decoder phase mask, the corresponding output of $i$th user's pulse from $j$th decoder is a low 
intensity pulse and has a pseudonoise behavior. In general, it can be demonstrated that when $i$th user sends pulse 
with intensity $I_i(t_0)=I_i$ at time $t_0$, its corresponding intensity at the $j$th output side is
%----------------------------------------------<<eq7>>--------------------------------------------------------
\begin{equation}
\ O_{j,i}(t)= \left\{
    \begin{array}{ll}
        \ \frac{I_i}{M} &  j=i \\
        \ \frac{I_i}{MN_0} & \mbox{other,}
    \end{array}
\right.
\end{equation}
where $N_0$ is the number of mask chips (also known as code length). The proof is similar to [17].

 To obtain channel statistics, we need to specify channel distribution. It can be shown that for a coherent CDMA channel with $M$ users, channel distribution can be formulated as follows
%----------------------------------------------<<eq7>>--------------------------------------------------------
\begin{equation}
\begin{split}
P(&y_j|  \forall  i  \in  \{ 1,2,... , M\} ,I_i(t_0)=I_i)\\
&=\frac{1}{\sum _{i=1, i \neq j}^{M-1}O_{j,i}}e^{-\frac{(y_j+ O_{j,j})}{\sum _{i=1, i \neq j}^{M-1}O_{j,i}}} I_0\left(\frac{2\sqrt{y_j O_{j,j}}}{\sum _{i=1, i \neq j}^{M-1}O_{j,i}}\right)
\label{qashsum}
\end{split}
\end{equation}
where $y_j$ is $j$th output data and $I_0(.)$ denotes modified Bessel function of first kind and zeroth order.

Assuming that $M$ is large enough, by using weak law of large numbers (WLLN), we can write
%---------------------------------------------<<eq7.1>>-------------------------------------------------------
\begin{equation}
\frac{\sum _{i=1, i \neq j}^{M-1}O_{j,i}}{(M-1)}=\mathbb{E} [O_{j,i}]=\frac{\mathbb{E} [I_i]}{MN_0}
\label{qashqash}
\end{equation}
Then equation (\ref{qashsum}) can be simplified to
%---------------------------------------------<<eq8>>-------------------------------------------------------
\begin{equation}
P(y_{j}|I_j,M-1 )=\frac{1}{{\beta\eta}}e^{-\frac{(y_{j}+\frac{I_j}{M})}{{\beta\eta}}}I_0\left(\frac{2\sqrt{y_j\frac{I_j}{M}}}{{\beta\eta}}\right)
\label{qashqash}
\end{equation}
where $\eta=\mathbb{E}[I_i(t)]$ and where $\beta=\frac{(M-1)}{MN_0}$. For more information about coherent CDMA  optical systems see [17].

Note that above formulations are derived without considering photo detector at the receiver end. If we take photo detector properties into account, (\ref{qashqash}) is transformed to
%----------------------------------------------<<eq9>>-------------------------------------------------------
\begin{equation}
\begin{split}
&W(y|x)\\
&=\frac{e^{\frac{-x}{M\beta\eta}}}{1+{\beta\eta}}\left(\frac{\beta\eta}{1+{\beta\eta}} \right)^y \left( \sum_{j=0}^{\infty} \left( \frac{x}{{M\beta\eta} (1+{\beta\eta})} \right)^j \frac{(y+j)!}{(j!)^2 y!} \right)
\end{split}
\label{laguerre}
\end{equation}
%%%%%%%%%%%%%%%%%%%%%%%%%%%%%%%%%%%%%%%%%
%

In this paper, the derivation of lower bounds is based on the data processing inequality [14]. The main idea in derivation of lower bounds relies on using an specific arbitrary distribution for input. In fact, by canceling optimization on input distribution, a lower bound can be obtained. By choosing $Q(\cdot)$ distribution for input, we have [15]
%--------------------------------------------<<eq6>>----------------------------------------------
\begin{equation}
\begin{split}
\mathcal{C}(A,\mathcal{E})=\stackrel[Q(\cdot)]{}{\sup} I(X;Y) \geq I(X;Y),%|_{\text{for a spacific $Q(.)$}},
\end{split}
\label{alpha}
\end{equation}

The derivation of the upper bound relies on the upper bound introduced in [16] as an upper bound for Discrete-Time Poisson channel. First we show that our Laguerre channel is degraded version of the Poisson channel without dark current, i.e. $\lambda_0=0$. Therefore, the upper bound that is obtained for Poisson channel can also be used for Laguerre channel.

%--------------------------------<< Results >>------------------------------------
\section{Results}

The results in this paper are proposed separately for short distance visible light communication systems and optical coherent CDMA networks in two subsections.

\subsection{short distance visible light communication systems}

In this subsection, we discuss lower and upper bounds on channel capacity under three cases: First case is in presence of average power and peak power constraints and $\alpha \in (0,1/3)$, second case is in presence of average power and peak power constraints and $\alpha \in [1/3,1]$ and third case is in presence of only average power constraint.

\subsubsection{Bounds on channel capacity in presence of average power and peak power constraints and $\alpha \in (0,1/3)$}
\begin{theo}
For short distance visible light communication systems, if $\alpha \in (0,\frac{1}{3})$, then channel capacity can be lower-bounded as
%----------------------------------------------<<eq0700>>--------------------------------------------------------
\begin{equation}
%\boxed{
\begin{split}
\mathcal{C}(A,\mathcal{E}) \geq &  \frac{1}{2}\log(A) -(1-\alpha)\mu -\log \left( \frac{1}{2}-\alpha \mu \right) \\
&-e^{\mu}\left(\frac{1}{2}-\alpha \mu \right)\vast[2\sqrt{\frac{12\lambda (\lambda +1)+1}{12A(2\lambda +1)}}\\
&\ \ \ \ \times \arctan\left(\sqrt{\frac{12A(2\lambda +1)}{12\lambda (\lambda +1)+1}} \right)\\
&+\log\left(1+\frac{12\lambda (\lambda +1)+1}{12A(2\lambda +1)}\right) \vast] \! - \! \frac{1}{2}\log(2\lambda+1)\\
&+\log(1+\frac{1+\lambda}{\mathcal{E}}) + (\mathcal{E}+\lambda)\log(1+\frac{1}{\mathcal{E} +\lambda}) \\
&-\frac{1}{2}\log 2\pi e -1,
\label{gagaga}
\end{split}
% }
\end{equation}
and upper-bounded as
\begin{equation}
%\boxed{
\begin{split}
\mathcal{C}(A,\alpha A) \leq & \frac{1}{2} \log A - (1-\alpha)\mu-\log\left( \frac{1}{2}-\alpha\mu \right)\\
&-\frac{1}{2}\log 2\pi e+o_A(1),
\end{split}
% }
\end{equation}
\end{theo}
where $o_A(1)$ denoted a function that tends to zero as its argument tends to infinity.
\subsubsection{Bounds on channel capacity in presence of average power and peak power constraints and $\alpha \in [1/3,1]$}
\begin{theo}
For short distance visible light communication systems, if $\alpha \in [\frac{1}{3},1]$, then channel capacity can be lower-bounded as
%----------------------------------------------<<eq0701>>--------------------------------------------------------
\begin{equation}
\begin{split}
\mathcal{C}(A,\mathcal{E}) \geq & \frac{1}{2}\log A+\log(1+\frac{3(1+\lambda)}{A})-1-\frac{1}{2}\log(\frac{\pi e}{2})\\
&+(\frac{A}{3}+\lambda)\log(1+\frac{3}{A+3\lambda})-\frac{1}{2}\log(2\lambda+1)\\
&-\sqrt{\frac{12\lambda (\lambda +1)+1}{12A(2\lambda +1)}} \arctan\left(\sqrt{\frac{12A(2\lambda +1)}{12\lambda (\lambda +1)+1}} \right)\\
&-\frac{1}{2}\log\left(1+\frac{12\lambda (\lambda +1)+1}{12A(2\lambda +1)}\right)
\end{split}
\label{agagag}
\end{equation}
and upper-bounded as
%----------------------------------------------<<eq0702>>--------------------------------------------------------
\begin{equation}
\begin{split}
\mathcal{C}(A,\alpha A) \leq \frac{1}{2}\log A - \frac{1}{2}\log(\frac{\pi e}{2})+o_A(1)
\end{split}
\end{equation}
\end{theo}
Note that in this case average power constraint is inactive, since for the distribution that maximize $h(X)-\frac{1}{2}E[\log X]$ when $\alpha \in [\frac{1}{3},1]$ we will have $E[X]=\frac{A}{3}$, therefore users use less than the allowed average power. Also, it is noteworthy to mention that lower bound (\ref{agagag}) can be obtained from lower bound (\ref{gagaga}) by tending $\alpha$ to $\frac{1}{3}$.
\subsubsection{Bounds on channel capacity in presence of only average power constraint}
\begin{theo}
For short distance visible light communication systems, if only average power constraint is imposed, then channel can be lower-bounded as
%----------------------------------------------<<eq0703>>--------------------------------------------------------
\begin{equation}
\begin{split}
\mathcal{C}(\mathcal{E}) \geq& \frac{1}{2}\log \mathcal{E}-\sqrt{\frac{12\pi\lambda (\lambda +1)+\pi}{24\mathcal{E}(2\lambda +1)}}+\log (1+\frac{1+\lambda}{\mathcal{E}})\\
&+(\mathcal{E}+\lambda)\log (1+\frac{1}{\mathcal{E}+\lambda})-\frac{1}{2}\log(2\lambda+1)-1
\end{split}
\end{equation}
and upper-bounded as
%----------------------------------------------<<eq0704>>--------------------------------------------------------
\begin{equation}
\begin{split}
\mathcal{C}(\mathcal{E}) \leq \frac{1}{2}\log \mathcal{E}+o_{\mathcal{E}}(1)
\end{split}
\label{agha}
\end{equation}
\end{theo}
Lower bound (\ref{agha}) can also be obtained from lower bound (\ref{gagaga}) when $\mathcal{E}$ is kept fix and $\alpha$ tends to zero.
\subsection{optical coherent CDMA networks}
%%%%%%%%%%%%%%%%%%%%%%%%%%%%
As mentioned before, in contrast to short distance visible light communication, in coherent optical CDMA networks the average noise power directly depends on the average power of input signals. One evident consequence of this dependency is that, by increasing the authorized average power for input signals, noise effect increases too, and it might cause the left hand side (LHS) of lower bound (\ref{gagaga}) to decrease instead of increasing. Therefore, the optimum $\alpha$ that maximizes LHS of (\ref{gagaga}), which is denoted by $\alpha^*$, might be different from one that found in previous section.  In other words, the optimum average power that maximizes LHS of (\ref{gagaga}) may differ from maximum authorized average power $\mathcal{E}$ and as a result, $\alpha^* \leq \frac{\mathcal{E}}{A}=\alpha$ . It is also possible that optimum $\alpha$ depends on maximum peak power constraint $A$. In order to find $\alpha ^*$ we should solve following equation
\begin{equation}
\begin{split}
\frac{\partial\left({\mathcal{C}(A,\mathcal{E})_l}\right)}{\partial{\alpha}}=0 \Rightarrow
\frac{\partial\left({\mathcal{C}(A,\mathcal{E})_l}\right)}{\partial{\mu}}=0
\end{split}
\label{hopeso}
\end{equation}
Equation (\ref{hopeso}) can be solved numerically. In the case in which there is not $\alpha$ such that $\frac{ \partial \left( \mathcal{C} (A,\mathcal{E})_l \right)}{{\partial \mu}}=0$ we set $\alpha ^* =\frac{1}{3}$. Fig. \ref{FigN} depicts $\alpha^*$ versus $A$ for some different values of $\beta$ and $M$.
\begin{figure}[t]
\centering
\includegraphics[scale=.5]{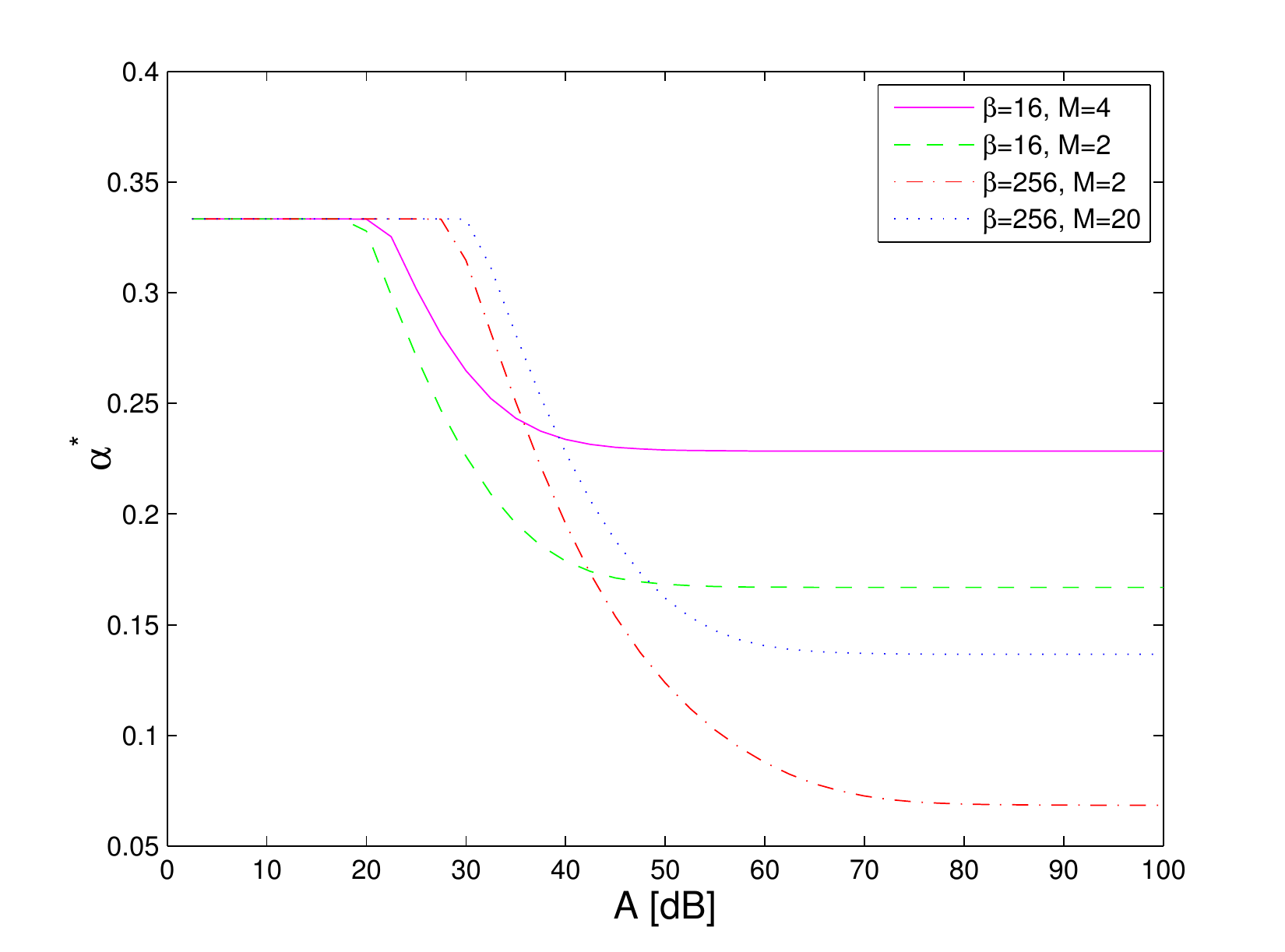}
\caption{Bound on the capacity of short distance visible light communication systems when both peak and average power constraints are active.}
\label{FigN}
\end{figure}
%%%%%%%%%%%%%%%%%%%%%%%%%%%%%%%%%%%%%%%%%%%%%%%%%%%%%%%%
However, one can propose a lower bound for coherent optical CDMA networks similarly to the one proposed for short 
distance visible light communication by adding maximization of the $\alpha$ to it. Thus, in the rest of this section, 
a lower bound for two cases of only peak power constraint and both average and peak power constraints are presented, 
then, by using resulted lower bound, a lower bound is proposed for the case that only average power constraint is 
imposed, and finally, this section ends with a brief discussion about sum-capacity in optical coherent CDMA networks.
%%%%%%%%%%%%%%%%%%%%%%%%%%
\subsubsection{Bounds on channel capacity in presence of average power and peak power constraints and $\alpha \in (0,1/3]$}
\begin{theo}
For coherent CDMA optical systems, if $\alpha \in (0,\frac{1}{3}]$, then channel capacity can be lower-bounded as
\begin{equation}
\begin{split}
\mathcal{C}(A,\mathcal{E}) \! \geq & \! \stackrel[0< \alpha \leq \frac{\mathcal{E}}{A} ]{}{\max} \! \vast\{ \frac{1}{2}\log(A) -(1-\alpha)\mu -\log \left( \frac{1}{2}-\alpha \mu \right) \\
&-e^{\mu}\left(\frac{1}{2}-\alpha \mu \right)\vast[2\sqrt{\frac{M\! \left( 12\beta \mathcal{E} (\beta \mathcal{E} +1)+1 \right)}{12A(2\beta \mathcal{E} +1)}}\\
&\ \ \ \ \times \arctan\left(\sqrt{\frac{12A(2\beta \mathcal{E} +1)}{M\! \left( 12\beta \mathcal{E} (\beta \mathcal{E} +1)+1 \right)}} \right)\\
&+\log\left(1+{\frac{M\! \left( 12\beta \mathcal{E} (\beta \mathcal{E} +1)+1 \right)}{12A(2\beta \mathcal{E} +1)}}\right) \vast]\\
&+\log(\frac{1}{M}+\frac{1+{\beta\mathcal{E}}}{\mathcal{E}})-\frac{1}{2}\log 2\pi e -1  \\
&-\frac{1}{2}\log(\frac{2{\beta\mathcal{E}} \! + \! 1}{M})+ (\frac{\mathcal{E}}{M} \! + \! {\beta\mathcal{E}})\log(1 \!+\! \frac{1}{\frac{\mathcal{E}}{M} \! + \! {\beta\mathcal{E}}}) \! \vast\}
\end{split}
% }
\end{equation}
%and upper-bounded as
%%%%%%%%%%%%%%%%%%%%%%%%%%%%%%%%%%%%%%%%%%%%%%%%%%%%%%%%

%\begin{equation}
%\boxed{
%\begin{split}
%\mathcal{C}(A,\mathcal{E}) \leq & \frac{1}{2} \log A - (1-\alpha)\mu-\log\left( \frac{1}{2} \log-\alpha\mu \right)\\
%&-\frac{1}{2}\log 2\pi e
%\end{split}
% }
%\end{equation}
\end{theo}
\subsubsection{Bounds on channel capacity in presence of only average power constraint}
\begin{theo}
For coherent CDMA optical systems, if only average power constraint is imposed, then channel capacity can be lower-bounded as
%%%%%%%%%%%%%%%%%%%%%%%%%%%%%%%%%%%%%%%%%%%%%%%%%%%%%%%%%%
%----------------------------------------------<<eq0703>>--------------------------------------------------------
\begin{equation}
\begin{split}
\mathcal{C}(\mathcal{E}) \geq& \frac{1}{2}\log \mathcal{E}-\sqrt{\frac{M\pi \left( 12\beta \mathcal{E} (\beta \mathcal{E} +1)+1 \right)}{24\mathcal{E}(2\beta \mathcal{E} +1)}}-1\\
&+(\frac{\mathcal{E}}{M}+{\beta\mathcal{E}})\log (1+\frac{1}{\frac{\mathcal{E}}{M}+{\beta\mathcal{E}}})\\
&+\log (\frac{1}{M}+\frac{1+{\beta\mathcal{E}}}{\mathcal{E}})-\frac{1}{2}\log(\frac{2{\beta\mathcal{E}}+1}{M})
\end{split}
\end{equation}
%and upper-bounded as
%%----------------------------------------------<<eq0704>>--------------------------------------------------------
%\begin{equation}
%\begin{split}
%\mathcal{C}(\mathcal{E}) \leq \frac{1}{2}\log \mathcal{E}+o_{\mathcal{E}}(1)
%\end{split}
%\end{equation}
%%%%%%%%%%%%%%%%%%%%%%%%%%%%%%%%%%%%%%%%%%%%%%%%%%%%%%%%%%%%%
\end{theo}
%%%%%%%%%%%%%%%%%%%%%%%%%%%%%%%%%%%%%%%%%%%%%%%%%%%%%%%%%%
\begin{figure}[t]
\centering
\includegraphics[scale=.5]{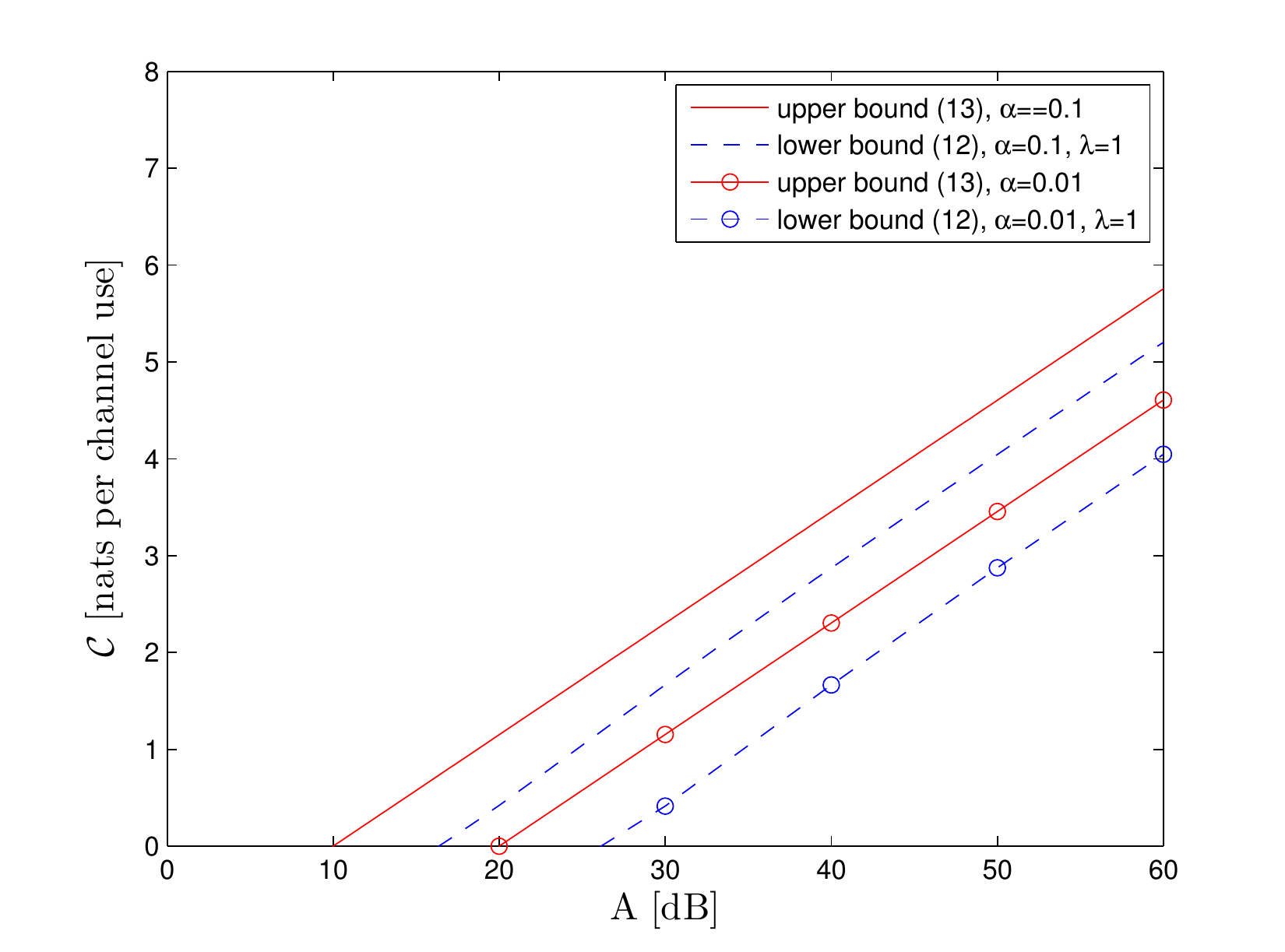}
\caption{Bound on the capacity of short distance visible light communication systems when both peak and average power constraints are active.}
\label{Fig1}
\end{figure}
\begin{figure}[t]
\centering
\includegraphics[scale=.5]{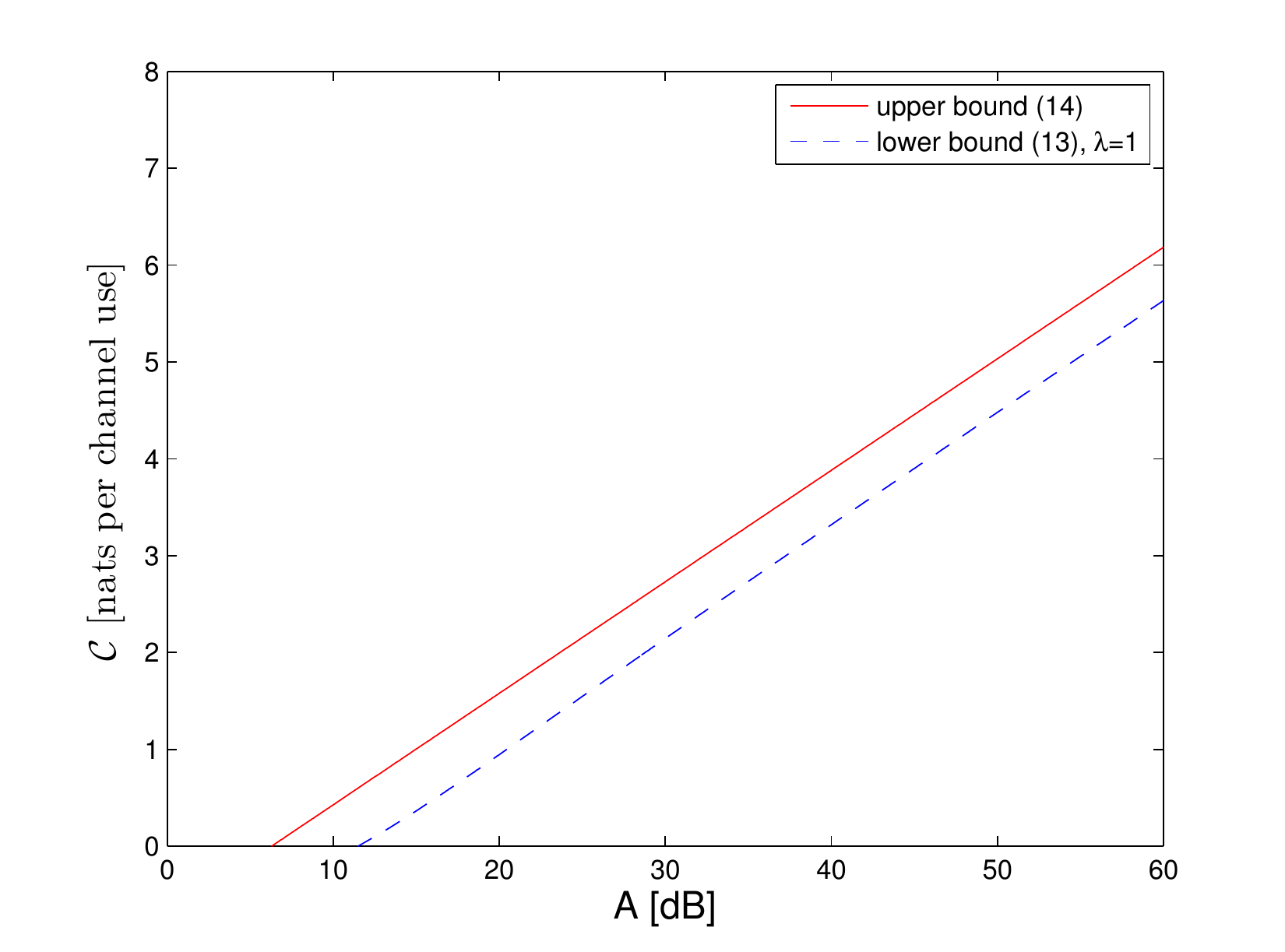}
\caption{Bound on the capacity of short distance visible light communication systems when only peak power constraint is active.}
\label{Fig2}
\end{figure}
\begin{figure}[h]
\centering
\includegraphics[scale=.5]{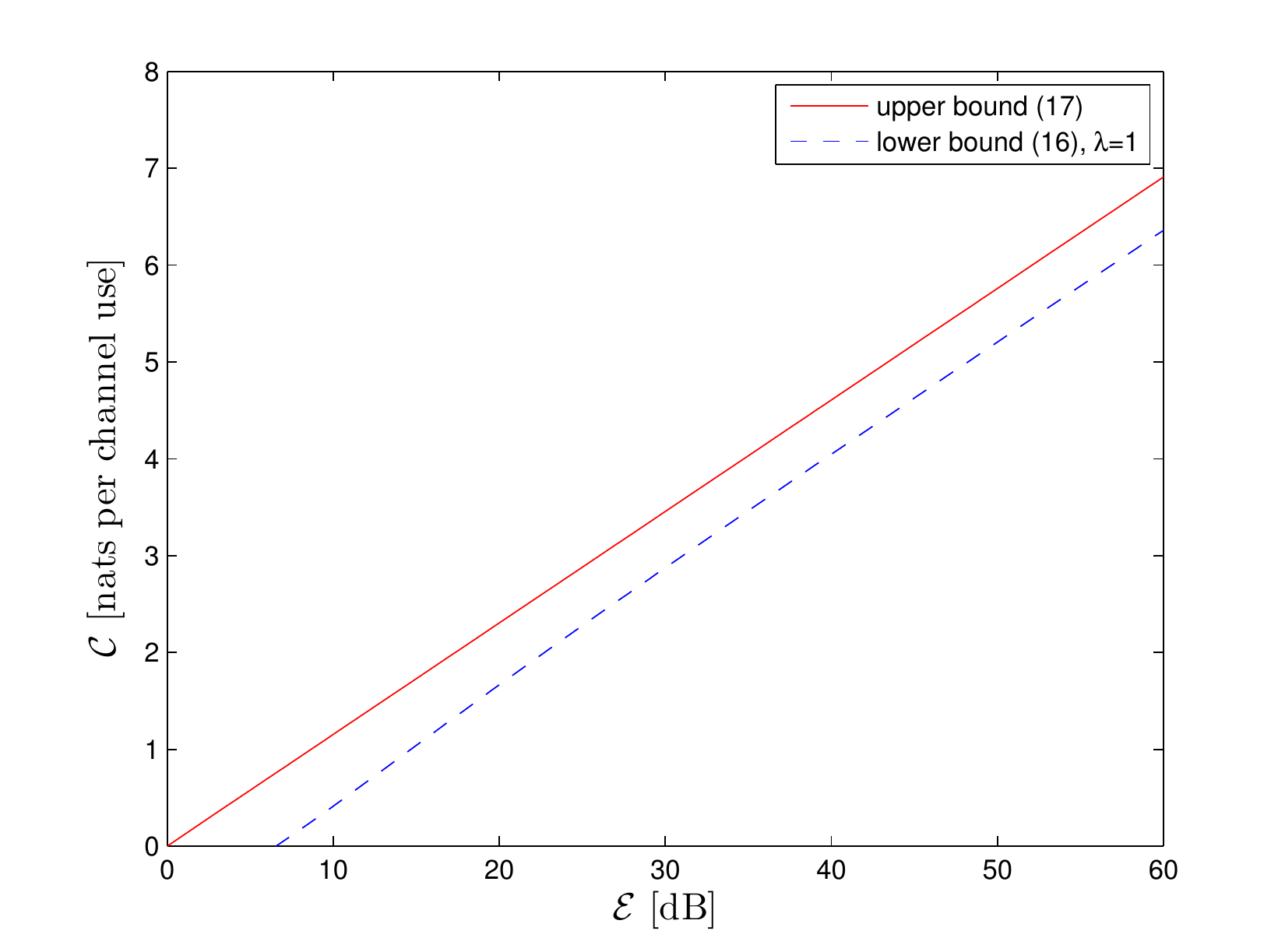}
\caption{Bound on the capacity of short distance visible light communication systems when only average power constraint is active.}
\label{Fig3}
\end{figure}
\begin{figure}[h]
\centering
\includegraphics[scale=.5]{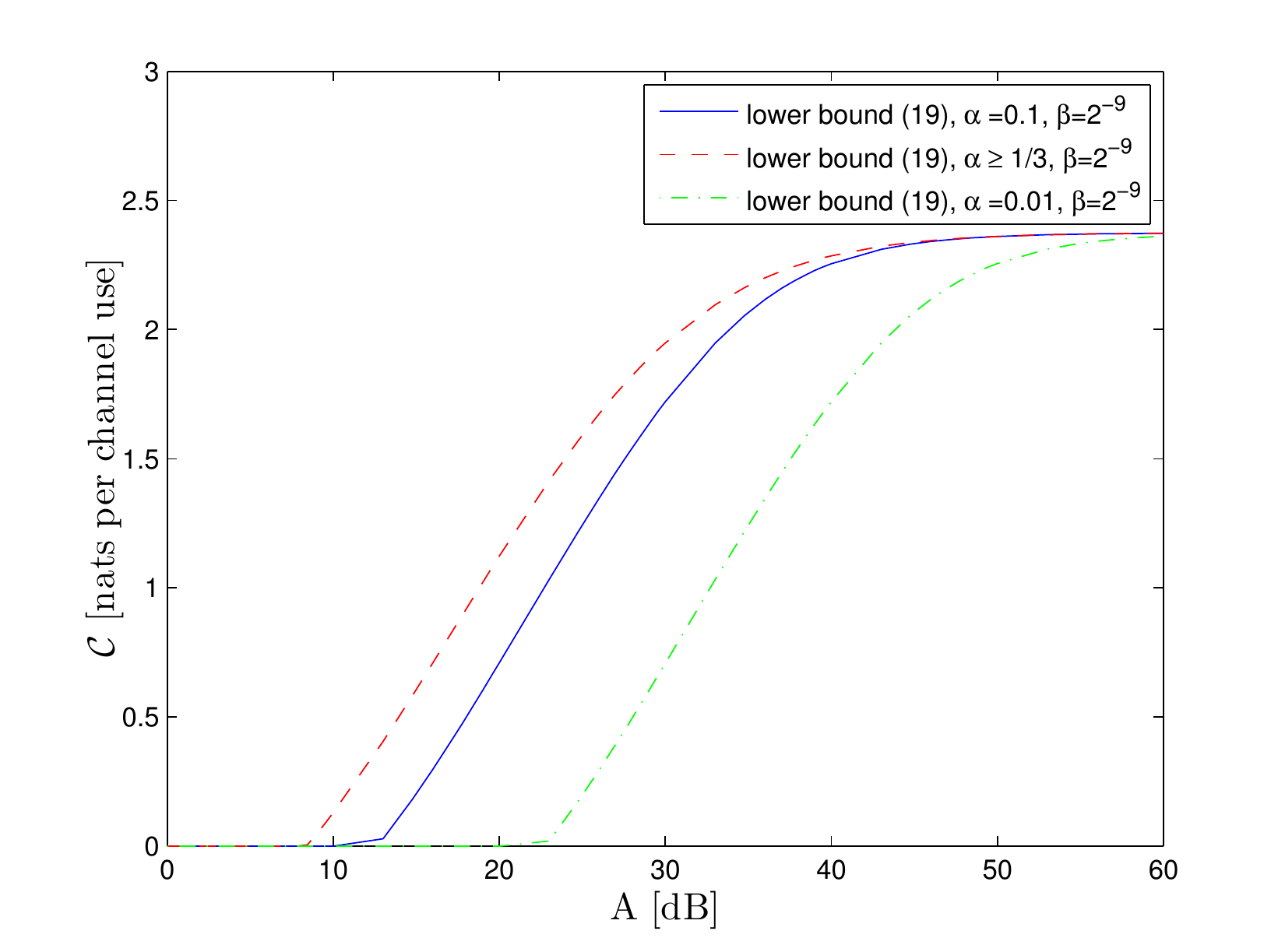}
\caption{Bound on the capacity of coherent CDMA systems when peak power constraint is active.}
\label{Fig4}
\end{figure}
\begin{figure}[h]
\centering
\includegraphics[scale=.5]{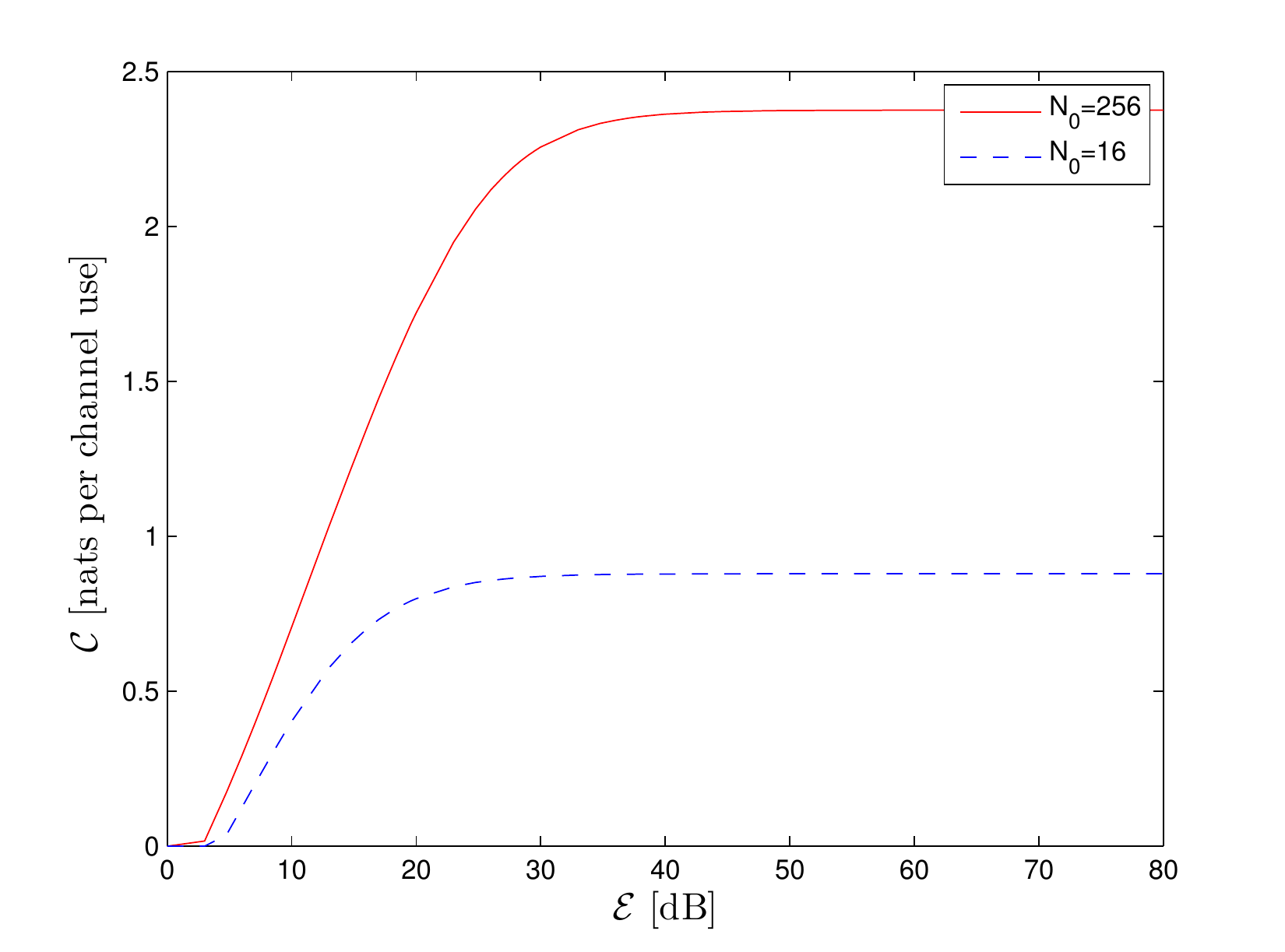}
\caption{Bound on the capacity of coherent CDMA systems when average power constraint is active.}
\label{Fig5}
\end{figure}
\begin{figure}[h]
\centering
\includegraphics[scale=.5]{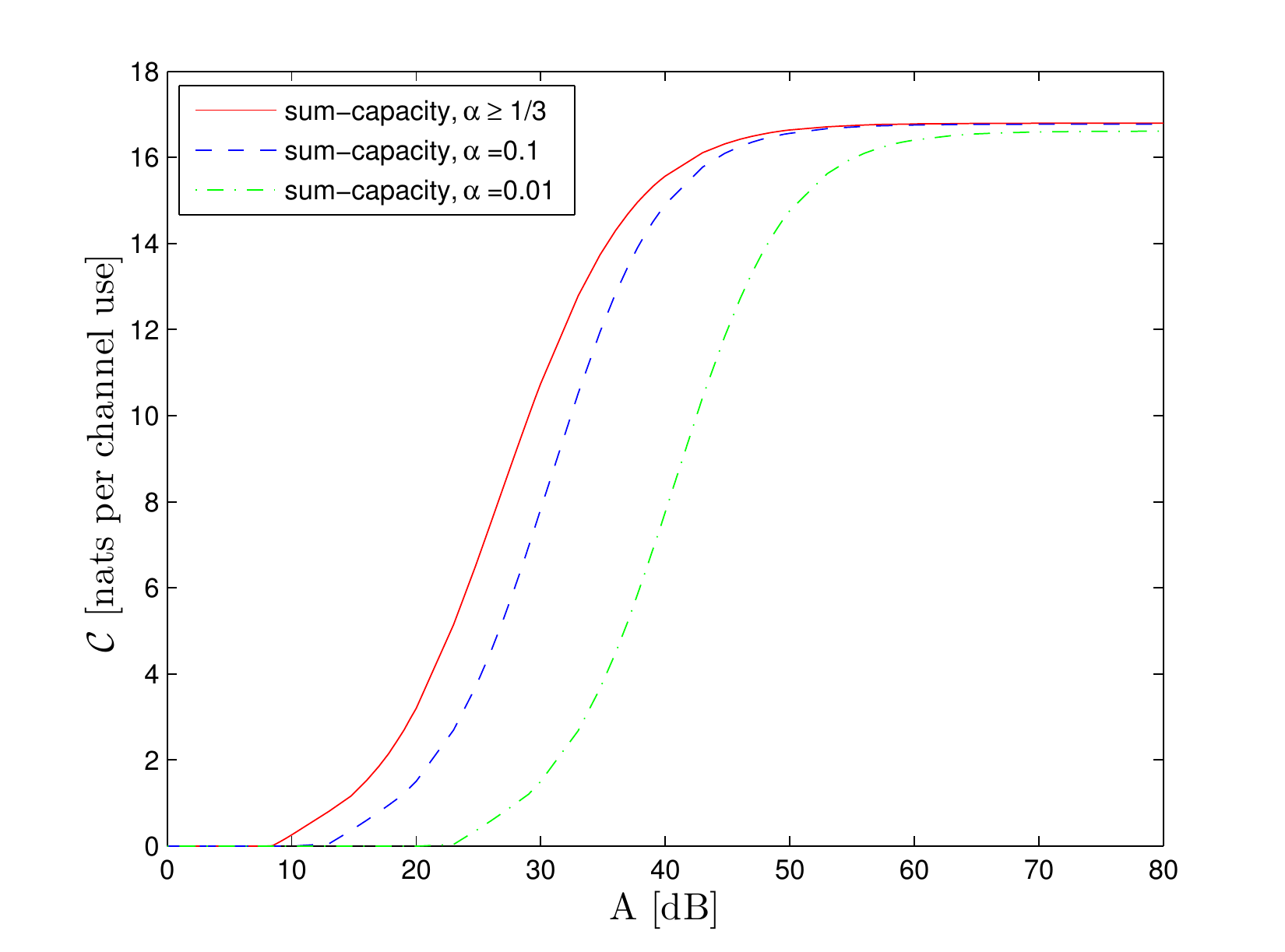}
\caption{Bound on the sum-capacity of coherent CDMA systems when peak power constraint is active.}
\label{sumcapacity6}
\end{figure}
\subsubsection{Evaluation of the sum-capacity}
As mentioned above, coherent CDMA channel is an interference channel and by considering the fact that each user sends its data independently of all other users, there is no cooperation among users. As a result, the interference channel treats as $M$ point-to-point single channels. In such a channel the average and peak input power of each user is $s$ and $s$ respectively and also the average input power of the noise is $s$ for every single channel due to the power constraints, therefore, the sum-capacity can be expressed as
\begin{equation}
\mathcal{C}_{sum}(\mathcal{A},\mathcal{E},M)=M\mathcal{C}(\mathcal{A},\mathcal{E})
\end{equation}
It is important to note that $\mathcal{C}(\mathcal{A},\mathcal{E})$ is a decreasing function of $M$, thus, one can conclude that for any given $\mathcal{A}$ and $\mathcal{E}$, there is an optimum $M$ which maximize sum-capacity. Fig. \ref{sumcapacity6} illustrates sum-capacity for three different cases of constraints discussed before and the optimum number of users for these three cases are depicted in Fig. \ref{mmm5}.
%%%%%%%%%%%%%%%%%%%%%%%%%%%%%%%%%%%%%%%%%%%%%%%%
%\begin{figure}[h]
%\centering
%\includegraphics[scale=1]{32}
%\caption{Bound on the capacity of coherent CDMA systems when only peak power constraint is active.}
%\label{Fig5}
%\end{figure}
%%
\begin{figure}[h]
\centering
\includegraphics[scale=.5]{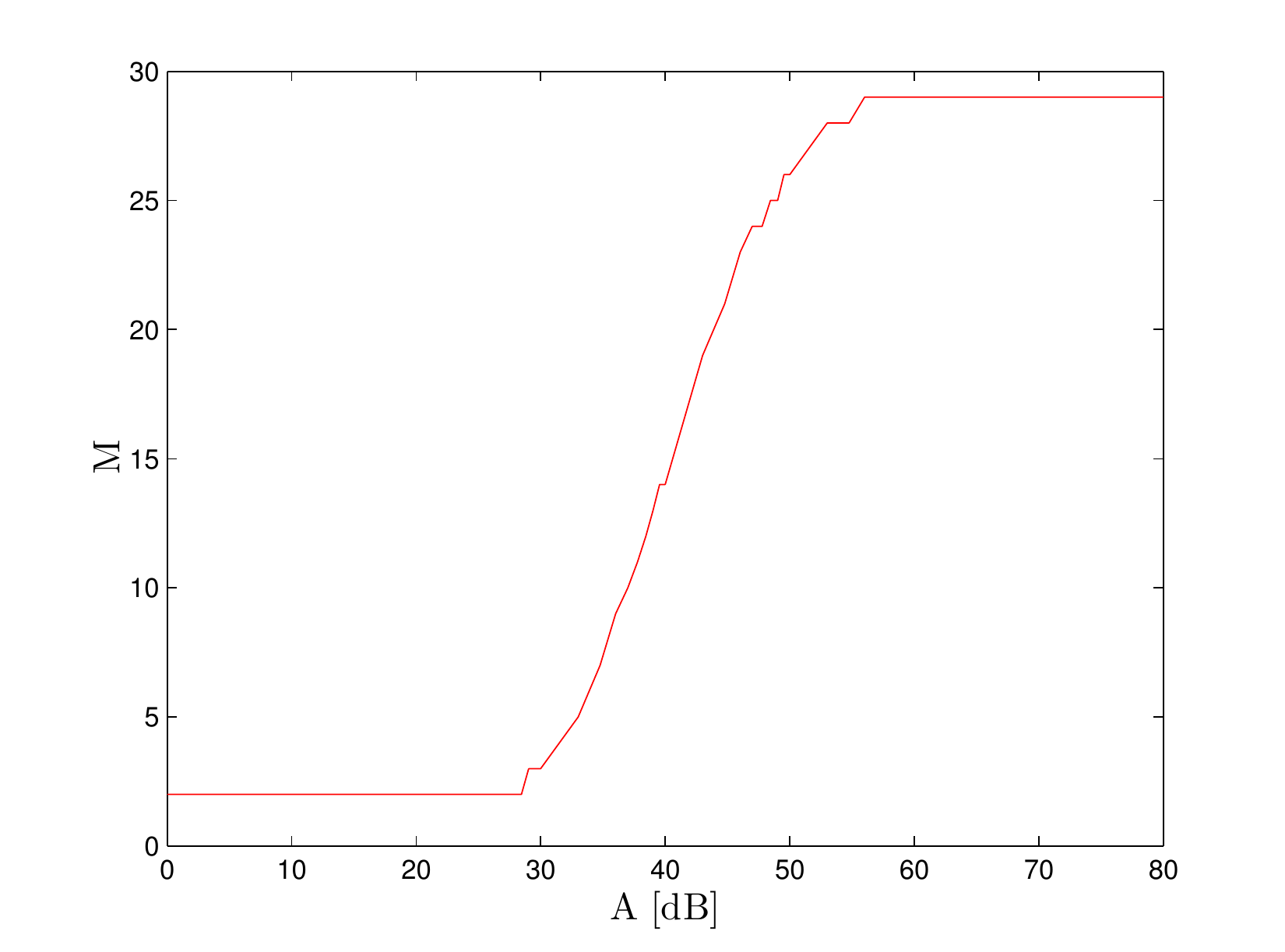}
\caption{Optimum number of users which maximaze sum-capacity.}
\label{mmm5}
\end{figure}
%%%%%%%%%%%%%%%%%%%%%%%%%%%%%%%%%%%%%%%%%%%%%%%%%
%--------------------------------<< Derivation >>------------------------------------
\section{Derivation}

In this section, derivations of the upper and lower bounds obtained in previous section are presented.

As mentioned before, for proving the lower bounds, one can drop the maximization on the input probability distribution as pointed in (6) and choose an arbitrary distribution in order to calculate the mutual information. But, since we are looking for tight lower bounds, input probability distribution must be chosen in a way that the corresponding mutual information results be as close as possible to capacity. It should be noted that not only obtaining such a distribution itself is a difficult problem, but also computing the entropy and distribution of its corresponding output might be indomitable. So, in order to avoid these issues, we apply a lower bound for $H(Y)$ in terms of $h(x)$ by using data processing theorem and an upper bound for $H(Y|X)$ with the help of following lemma:
\begin{lemma}
Let $K$ be a Laguerre random variable with distribution
\begin{equation}
\begin{split}
Pr[K]=&\frac{e^{\frac{-I_p(t)}{\lambda}}}{1+\lambda}\left(\frac{\lambda}{1+\lambda} \right)^K \left( \sum_{j=0}^{\infty} \left( \frac{I_p(t)}{\lambda (1+\lambda)} \right)^j \frac{(K+j)!}{(j!)^2 K!} \right)\\
=&\frac{e^{\frac{-I_p(t)}{1+\lambda}}}{1+\lambda}\left(\frac{\lambda}{1+\lambda} \right)^K \\
&\times  \left( \sum_{j=0}^{K} \left( \frac{I_p(t)}{\lambda (1+\lambda)} \right)^j \frac{K!}{(j!)^2 (K-j)!} \right)
\label{kkkkklab}
\end{split}
\end{equation}
then we can upper-bound $H(K)$ as below
%--------------------------------------------------------<<eq10>>------------------------------------------------
\begin{equation}
H(K) \leq \frac{1}{2}\log 2\pi e (\eta(1+2\lambda) +\lambda (1+\lambda)+ \frac{1}{12}),
\end{equation}
where $\eta=E[I_p(t)]$.
\end{lemma}
\begin{proof}
With a bit of mathematical analysis, it can be shown that the variance of distribution (\ref{kkkkklab}) is
\begin{equation}
Var (K)=\eta(1+2\lambda) +\lambda (1+\lambda),
\end{equation}
then the proof completes by using [2, Theorem 16.3.3].
\end{proof}
By using data processing theorem we can write
%--------------------------------------------------------<<eq11>>------------------------------------------------
\begin{equation}
D(Q(.)||Q_E(.)) \geq D(R(.)||R_G(.)),
\end{equation}
where $Q(.)$ denotes an arbitrary CDF with mean $\eta$ on $\mathbb{R}^+$ and $R(.)$ denotes its corresponding output on $\mathbb{Z}^+$ when channel statistics is laguerre. Similarly, $Q_E(.)$ denotes exponential CDF with mean $\eta$ on $\mathbb{R}^+$ and $R_G(.)$ denotes its corresponding output on $\mathbb{Z}^+$. It cab be shown that $R_G(.)$ is a geometric PMF with mean $\eta + \lambda$.
\newline
(See Appendix for the proof)

for the left hand side of () we have
%--------------------------------------------------------<<eq12>>------------------------------------------------
\begin{equation}
\begin{split}
D(Q(.)||Q_E(.)) =& \int_0^\infty Q'(x) \log \frac{Q'(x)}{\frac{1}{\eta}e^{\frac{-x}{\eta}}}\\
=&-h(X)+\log (\eta) +1
\end{split}
\end{equation}
and for the right-hand side one can show that
%--------------------------------------------------------<<eq13>>------------------------------------------------
\begin{equation}
\begin{split}
D(R(.)||R_G(.))=&\sum^{\infty}_{y=0} R(y)\log{\frac{R(y)}{\frac{1}{1+\eta+\lambda}\left(\frac{\eta +\lambda}{1+\lambda+\eta} \right)^y}}\\
=&-H(Y)+(1+\lambda +\eta)\log(1+\lambda +\eta)\\
 &- (\eta +\lambda)\log(\eta +\lambda)
\end{split}
\end{equation}
and finally, we have
%--------------------------------------------------------<<eq14>>------------------------------------------------
\begin{equation}
\begin{split}
%&D(Q(.)||Q_E(.)) \geq D(R(.)||R_G(.)) \Rightarrow \\
%&-h(X)+\log(\eta) +1 \geq -H(Y)+(1+\lambda +\eta)\log(1+\lambda +\eta) - (\eta +\lambda)\log(\eta +\lambda) \Rightarrow\\
H(Y) \geq& h(X) +\log(1+\frac{1+\lambda}{\eta}) + (\eta+\lambda)\log(1+\frac{1}{\eta +\lambda}) \\
&-1 %\Rightarrow \\
%&H(Y) \geq h(X) +\log(1+\frac{1+\lambda}{E[X]}) + (E[X]+\lambda)\log(1+\frac{1}{E[X] +\lambda}) -1 \ \ \ \ \ \square
\end{split}
\end{equation}
%Above equation follows from the fact that
%\begin{equation}
%\end{equation}

The remainder of the derivation of  lower bounds is based on maximizing differential entropy under the given constraints. To this goal, we choose CDF $Q(.)$ to maximize $h(X)-\frac{1}{2}E [\log X]$ either under constraints (2) or (3) or both. These distributions can be represented with the following densities:
\begin{equation}
Q'(x)=\frac{\sqrt{\mu}}{\sqrt{A\pi x}.erf(\sqrt{\mu})}.\stackrel[]{}{e^{-\frac{\mu}{A}x}},\ \ \ \ 0 \leq x \leq A
\end{equation}
when both constraints (2) and (3) are active, and
\begin{equation}
Q'(x)=\frac{1}{\sqrt{4Ax}},\ \ \ \ 0 \leq x \leq A
\end{equation}
when constraint (2) is inactive and (3) is active, and
\begin{equation}
Q'(x)=\frac{1}{\sqrt{2 \pi \mathcal{E}x}},\ \ \ \ x \geq 0
\end{equation}
when only constraint (2) is imposed.

Finally, lower bounds can be obtained by using above input distributions analogously to [4].

%%%%%%%%%%%%%%%%%%%%%%%%%%%%%%%%%%%%%
\begin{figure}[t]
\centering
\includegraphics[scale=1.1]{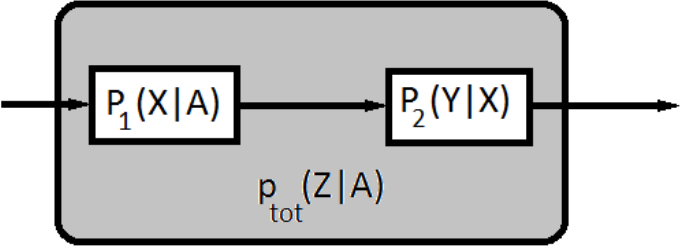}
\caption{A typical degradd channel}
\label{Fig7}
\end{figure}
%%%%%%%%%%%%%%%%%%%%%%%%%%%%%%%%%%%%%%%%%%%%%%%%%%%%

Derivation of upper bounds are based on data processing inequality. The proof is structured by following steps: at first, we will show that every Laguerre channels with density (1) and arbitrary average noise power $\lambda$ are degraded version of Poisson channel with no dark current. Then, with the help of Markov chain and data processing inequality, we will show that every upper bounds which is valid for Poisson channel with no dark current, is also valid for Laguerre channel with distribution (1), therefore, we can apply asymptotic upper bounds introduced in [6] to our model. We start the proof of the upper bounds with the following lemma:
\begin{lemma}
Consider the degraded channel depicted in the figure (\ref{Fig7}), if $p_2(y|x)$ can be expressed as below
\begin{equation}
p_2(y|x)=\sum_{(\sum_{i=1}^y y_i=y)} \prod_{i=1}^{x}p_2(y_i|1) \ \ \ y,x \in \mathbb{Z}^+,
\end{equation}
then we have
\begin{equation}
\Psi_{tot}(z)=\Psi_1(\Psi_2(z)),
\end{equation}
where $\Psi_1(z)$, $\Psi_2(z)$ and $\Psi_{tot}(z)$ are moment generating function of $p_1(x|a)$, $p_2(y|1)$ and $p_{tot}(y|a)$ respectively.
\end{lemma}
%%%%%%%%%%%%%%%%%%%%%%%%%%%%%%%%%%%%%%%%%%%%%%%%%%%%%%
\begin{figure}[t]
\centering
\includegraphics[scale=.3]{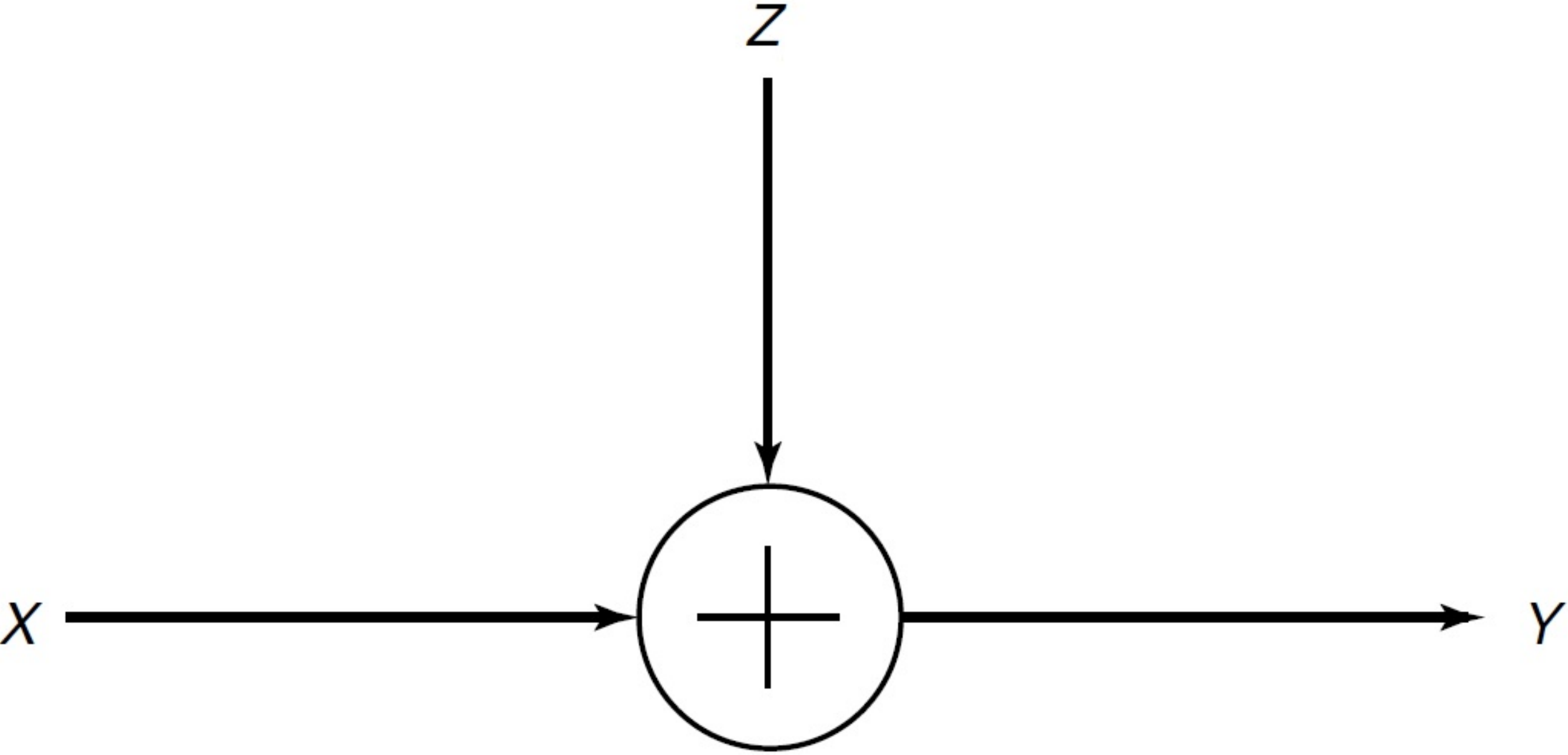}
\caption{A typical channel with independent noise}
\label{Fig8}
\end{figure}
%%%%%%%%%%%%%%%%%%%%%%%%%%%%%%%%%%%%%%%%%%%%%%%%%%%%%%%
\begin{proof}
\begin{equation}
\begin{split}
\Psi_{tot}(z)=&\sum_{y=0}^{\infty}p(y|a)z^y\\
=&\sum^{\infty}_{y=0}\left(p_2(y|x) \sum^{\infty}_{x=0}p_1(x|a)\right)z^y\\
=&\sum^{\infty}_{y=0}\left(\sum_{(\sum_{i=1}^x y_i=y)} \prod_{i=1}^{x}p_2(y_i|1) \sum^{\infty}_{x=0}p_1(x|a)\right)z^y\\
=&\sum^{\infty}_{x=0}\left(\sum^{\infty}_{y=0}\sum_{(\sum_{i=1}^x y_i=y)} \prod_{i=1}^{x}\left(p_2(y_i|1)z^{y_i}\right) \right)p_1(x|a)\\
=&\sum^{\infty}_{x=0}\left(\sum_{y_1=0}^{\infty} \sum_{y_2=0}^{\infty}... \sum_{y_x=0}^{\infty} \prod_{i=1}^{x}\left(p_2(y_i|1)z^{y_i}\right) \right)p_1(x|a)\\
=&\sum^{\infty}_{x=0}\vast(\left(\sum_{y_1=0}^{\infty}p_2(y_1|1)z^{y_1}\right)\left(\sum_{y_1=0}^{\infty}p_2(y_2|1)z^{y_2}\right)...\\
&\times \left(\sum_{y_1=0}^{\infty}p_2(y_x|1)z^{y_x}\right) \vast)p_1(x|a)\\
=&\sum_{x=0}^{\infty}p_1(x|a)\Psi_2(z)^x\\
=&\Psi_1(\Psi_2(z))
\end{split}
\end{equation}
\end{proof}
For such a channel by data processing inequality one can conclude that
%\begin{equation}
%\begin{split}
%\[
%I(A;Y) \leq I(A;X)
%\]
\begin{equation}
\begin{split}
&I(A;Y) \leq \min \left( I(X;Y),I(A;X)\right).
\end{split}
\end{equation}
In addition, it is obvious that by adding an independent noise to output of the above channel, as depicted in figure (\ref{Fig8}), one can make another version of  degraded channel. Then, it can be concluded that
\begin{equation}
\begin{split}
\Psi_{z}(z)=\Psi_{tot}(z).\Psi_n(z)
\end{split}
\end{equation}
where $\Psi_n(z)$ is moment generating function of the noise. Also, in this channel Because of the Markov chain we have
\begin{equation}
I(A;Z) \leq I(A;Y).
\end{equation}
The proof of the upper bounds finishes by following lemma:
\begin{lemma}
The Laguerre channel with arbitrary $\lambda$ is a degraded version of Poisson channel with no dark current.
\end{lemma}
\begin{proof}
Regarding to lemma 2, it is sufficient to show that Laguerre MGF can be expressed in terms of Poisson MGF. Considering $v_1$ and $v_2$ as a Birth-Death and Bose-Einstein process, we have
\begin{equation}
\Psi_{v_1}(z)=\frac{z+\lambda(1-z)}{1+\lambda(1-z)}
\end{equation}
and
\begin{equation}
\Psi_{v_2}(z)=\frac{1}{1+\lambda(1-z)},
\end{equation}
and it can be easily shown that
\begin{equation}
\Psi_{Laguerre}(z)=\Psi_{Poisson}(\Psi_{v_1}(z)).\Psi_{v_2}(z)
\end{equation}
%It can be easily shown that Laguerre MGF is cascade of Poisson MGF and Birth Death (B.D.) MGF hich is producted to an independent Bose-Einstein noise
The reminder of the proof is straightforward by equations (34) and (36). Therefore, we have
\begin{equation}
\mathcal{C}_{Laguerre}(A,\mathcal{E}) \leq \mathcal{C}_{Poisson}(A,\mathcal{E}) \leq {U}_{Poisson}(A,\mathcal{E})
\end{equation}
where $U_{Poisson}$ is the upper-bounds acquired in [6].
\end{proof}
%%%%%%%%%%%%%%%%%%%%%%%%%%%%%%%%%%%%%%%%%%%%%%%%%%%%%%%%%%
%%%%%%%%%%%%%%%%%%%%%%%%%%%%%%%%%%%%%%%%%%%%%%%%%%%%%%%%%%%
%---------------------------------------------------------<<eq>>-------------------------------------------------
\appendix

%%%%%%%%%%%%%%%%%%%%%%%%%%%%%%%%%%%%%%%%%%%%%%%%
%%%%%%%%%%%%%%%%%%%%%%%%%%%%%%%%%%%%%%%%%%%%%%%%%%%
%%%%%%%%%%%%%%%%%%%%%%%%%%%%%%%%%%%%%%%%%%%%%%%%
%%\begin{figure*}[bottom!]
%% ensure that we have normalsize text
%%\normalsize
%% Store the current equation number.
%%\setcounter{MYtempeqncnt}{\value{equation}}
%% Set the equation number to one less than the one
%% desired for the first equation here.
%% The value here will have to changed if equations
%% are added or removed prior to the place these
%% equations are referenced in the main text.
%%\setcounter{equation}{11}
%%\hrulefill
%%% The spacer can be tweaked to stop underfull vboxes.
%%\vspace
%%*
%%{4pt}
%%\begin{equation}
%%\begin{split}
%%\frac{\partial{\mathcal{C}_{l}(A,\mathcal{E})}}{\partial{\mu}}=\stackrel[]{}{}
%%\end{split}
%%\end{equation}
%% Restore the current equation number.
%%\setcounter{equation}{12}
%%\setcounter{equation}{\value{MYtempeqncnt}}
%% IEEE uses as a separator
%%
%%\end{figure*}
%%%%%%%%%%%%%%%%%%%%%%%%%%%%%%%%%%%%%%%%%%%%
%%%%%%%%%%%%%%%%%%%%%%%%%%%%%%%%%%%%%%%%%%%%%%%
%%%%%%%%%%%%%%%%%%%%%%%%%%%%%%%%%%%%%%%%%%%%%%

\section{}
In this Appendix, we prove that when the input of Laguerre channel is negative exponential distribution, the corresponding output distribution is Bose-Einstein distribution. By substituting We have

\begin{subequations} \label{eq:energy13}
\begin{align}
%\begin{split}
%&
&R_G(y)\nonumber\\
&=\int_{0}^{\infty}W(y|x)Q_E(x)dx\\
%&
&=\int_{0}^{\infty}\frac{e^{\frac{-x}{\lambda}}}{1+\lambda}\left(\frac{\lambda}{1+\lambda} \right)^y \nonumber \\
&\qquad \times  \left( \sum_{j=0}^{\infty}\left( \frac{x}{\lambda (1+\lambda)} \right)^j \frac{(y+j)!}{(j!)^2 y!} \right)\frac{1}{\eta}e^{\frac{-x}{\eta}}dx\\
%&
&=\frac{1}{\eta(1+\lambda)}\left( \frac{\lambda}{1+\lambda} \right)^y \left( \sum_{j=0}^{\infty} \frac{(y+j)!}{(j!)^2 (y!)} \right)\nonumber \\
&\qquad \times\left( \int_{0}^{\infty}\left( \frac{x}{\lambda(1+\lambda)} \right)^j e^{-x(\frac{1}{\eta}+\frac{1}{\lambda})} dx\right)\\
%&
&=\frac{1}{\eta(1+\lambda)}\left( \frac{\lambda}{1+\lambda} \right)^y \left( \sum_{j=0}^{\infty} \frac{(y+j)!}{(j!)^2 (y!)} \right)\nonumber \\
&\qquad \times \left( j! \left( \frac{1}{\lambda(1+\lambda)} \right)^j \left( \frac{\lambda \eta}{\lambda + \eta} \right) ^{j+1} \right) \label{sub-eq4} \\
%&
&=\frac{1}{\eta(1+\lambda)}\left( \frac{\lambda}{1+\lambda} \right)^y \left( \sum_{j=0}^{\infty} \frac{(y+j)!}{(j!) (y!)} \right)\nonumber \\
&\qquad \times\left( \frac{\eta}{(1+ \lambda)(\lambda +\eta)} \right)^j \left( \frac{\lambda \eta}{\lambda+\eta} \right)\\
%&
&=\frac{1}{\eta(1+\lambda)}\left( \frac{\lambda}{1+\lambda} \right)^y  \left( \frac{\lambda \eta}{\lambda+\eta} \right)\nonumber \\
&\qquad \times \left( \sum_{j=0}^{\infty} \binom{y+j}{j} \left( \frac{\eta}{(1+ \lambda)(\lambda +\eta)} \right)^j \right)\\
%&
&=\frac{1}{\eta(1+\lambda)}\left( \frac{\lambda}{1+\lambda} \right)^y  \left( \frac{\lambda \eta}{\lambda+\eta} \right) \left( \frac{(1+\lambda)(\lambda +\eta)}{\lambda (1+\lambda +\eta)} \right)^{y+1} \label{sub-eqg}\\
%&
&=\frac{1}{1+\lambda +\eta}\left( \frac{\lambda}{\lambda +1}\right)^{y+1} \left( \frac{\lambda+1}{\lambda}\right)^{y+1} \left(\frac{\eta +\lambda}{1+\lambda +\eta} \right)^y\\
%&
&=\frac{1}{1+\lambda +\eta} \left(\frac{\eta +\lambda}{1+\lambda +\eta} \right)^y
%\end{split}
\end{align}
\end{subequations}
%------------------------------------------------------------------------------------------------------------------------
Where (42d) follows from the fact that
%--------------------------------------------<<eqqq>>--------------------------------------------------------------
\begin{equation}
\int_{0}^{\infty}(ax)^ke^{-bx}=\frac{(a)^k (b)^{-k} \Gamma(k+1)}{b}=k!\frac{a^k}{b^{k+1}}
\end{equation}
%------------------------------------------------------------------------------------------------------------------------
and (42g) can be obtain by using
%--------------------------------------------<<eqqq>>--------------------------------------------------------------
\begin{equation}
\sum^{\infty}_{i=0}\binom{b+i}{i}a^i=\left(\frac{1}{1-a}\right)^{b+1}
\end{equation}
%------------------------------------------------------------------------------------------------------------------------
Thus, the proof is complete.

Similarly, it is easy to show that for optical coherent CDMA network by choosing input 
distribution $Q_E(x)=\frac{1}{\eta}e^{\frac{-x}{\eta}}$, the corresponding output distribution will 
be $R_G(y)=\frac{1}{\frac{1}{M}+\lambda +\eta} \left(\frac{\eta +\lambda}{\frac{1}{M}+\lambda +\eta} \right)^y$.
%%%%%%%%%%%%%%%%%%%%%%%%%%%%%%%%%%%%%%%%%%%%%%%%%%%%%%%%%%%%%
%%%%%%%%%%%%%%%%%%%%%%%%%%%%%%%%%%%%%%%%%%%%%%%%%%%%%%%%%%%%%%

\end{document}